\documentclass[journal,web]{ieeecolor}
\usepackage{etoolbox}
\makeatletter
\@ifundefined{color@begingroup}%
{\let\color@begingroup\relax
	\let\color@endgroup\relax}{}%
\def\fix@ieeecolor@hbox#1{%
	\hbox{\color@begingroup#1\color@endgroup}}
\patchcmd\@makecaption{\hbox}{\fix@ieeecolor@hbox}{}{\FAILED}
\patchcmd\@makecaption{\hbox}{\fix@ieeecolor@hbox}{}{\FAILED}
\usepackage{generic}
\usepackage{cite}
\usepackage{amsmath,amssymb,amsfonts}

\usepackage{multirow}
\usepackage{graphicx}
\usepackage{subfigure}
\usepackage{textcomp}
\usepackage{makecell}
\usepackage{color}
\usepackage{enumerate}

\usepackage{xcolor}
\usepackage{bm}
\usepackage{amstext}
\usepackage{stfloats}
\usepackage{subfigure}
\newtheorem{remark}{Remark}{}
\newtheorem{theorem}{Theorem}
\newtheorem{proposition}{Proposition}
\newtheorem{lemma}{Lemma}
\newtheorem{definition}{Definition}

\newtheorem{assumption}{Assumption}

\newtheorem{algorithm}{Algorithm}

\usepackage{mathrsfs}
\usepackage{diagbox}
\allowdisplaybreaks[4]
\usepackage{cases}
\usepackage{subeqnarray}
\def\BibTeX{{\rm B\kern-.05em{\sc i\kern-.025em b}\kern-.08em
T\kern-.1667em\lower.7ex\hbox{E}\kern-.125emX}}

	\usepackage{algorithm}
	\usepackage{algorithmicx}
	\usepackage{algpseudocode}
	
\begin{document}
	\title{Active Secure Neighbor Selection in Multi-Agent Systems with Byzantine Attacks}

	\author{Jinming~Gao,
	Yijing~Wang,
	Wentao~Zhang,~\IEEEmembership{Member,~IEEE},
	Rui~Zhao,~\IEEEmembership{ Member,~IEEE}, Yang~Shi,~\IEEEmembership{Fellow,~IEEE}, and  
	Zhiqiang~Zuo,~\IEEEmembership{Senior Member,~IEEE}
	\thanks{This work was supported by the National Natural Science Foundation of China
		No. 62173243.}
	\thanks{J.  Gao, Y.  Wang, and Z. Zuo are with the Tianjin Key Laboratory of Intelligent Unmanned Swarm Technology and System, School of Electrical and Information Engineering, Tianjin
		University, 300072, P. R. China.
		(e-mail: gjinming@tju.edu.cn; yjwang@tju.edu.cn; zqzuo@tju.edu.cn)}
	\thanks{R.~Zhao is with the Department of Electrical Engineering, City University of Hong Kong, Hong Kong SAR, China. (e-mail: ruizhao@tju.edu.cn, ruzhao@cityu.edu.hk)}
	\thanks{W. Zhang was with Continental-NTU Corporate Lab, Nanyang Technological University, 639798, Singapore, and will join the School of Robotics,
		Hunan University, Changsha 410082, China (e-mail: wtzhangee@tju.edu.cn, wentao.zhang@ntu.edu.sg)}	
	\thanks{Y. Shi is with the Department of Mechanical Engineering, University of Victoria, Victoria, BC V8W 2Y2, Canada. (e-mail: yshi@uvic.ca)}
}
	
	\maketitle
	
	\begin{abstract}
This paper investigates the problem of resilient control  for multi-agent systems in the presence of Byzantine adversaries via an active secure neighbor selection framework. A pre-discriminative graph is first constructed to characterize the admissible set of candidate neighbors for each agent. Based on this graph, a dynamic in-neighbor selection strategy is proposed, wherein each agent actively selects a subset of its pre-discriminative neighbors. The number of selected neighbors is adjustable, allowing for a trade-off between communication overhead and robustness, with the minimal case requiring only a single in-neighbor. The proposed strategy facilitates the reconstruction of a directed spanning tree among normal agents following the detection and isolation of Byzantine agents. It achieves resilient consensus without imposing any assumptions on the initial connectivity among normal agents. Moreover, the approach significantly reduces communication burden while maintaining resilience to adversarial behavior. A numerical example is provided to illustrate the effectiveness of the proposed method.
	\end{abstract}

\section{Introduction}
Significant progress in cyber-physical systems (CPSs) has been driven by advances in communication and computing technologies  \cite{103680,watari2023duck}. However, the open setting in cyber space poses security challenges for real-world deployments, as exemplified by the 2010 Stuxnet attack on Iran's nuclear facilities \cite{mo2015physical} and the 2014 Havex instrusion that disabled hydropower dams over SCADA networks \cite{Maitra2015}.

Given the high security vulnerability of CPSs, resilient defense mechanisms are paramount for ensuring normal system operation.  Substantial research has focused on  attack detection and identification  \cite{pasqualetti2013attack,Gallo2020}, while recent efforts increasingly emphasize  attack mitigation \cite{pirani2023graph}.   For example,  \cite{zhao2023active} proposed  an active switching approach to  defend against  denial of service attacks. To ensure resilient control in distributed systems, redundancy-based schemes have been developed by leveraging their structural characteristics \cite{pirani2023graph}.  In this way, redundant security components will be  used, such as sensors \cite{lu2023distributed} or communication  links \cite{leblanc2013resilient}, to achieve security estimation or resilient  control.

Multi-agent systems (MASs), as a prominent class of  CPSs, 
have  received enormous attention owing  to their widespread 
applications \cite{olfati2007consensus,chen2024cooperative}. Unlike the  centralized systems, MASs comprise 
multiple autonomous agents that can be sparsely distributed and easily scaled.  The applications include intelligent traffic systems \cite{10014016}, smart grid systems \cite{10018476} and multi-sensor networks \cite{zheng2018average}.

Yet, due to their scalability and complexity restrictions, MASs are intrinsically more susceptible to adversarial manipulation than  centralized systems  \cite{zhang2022much,zuo2021resilient}. Guaranteeing resilient consensus under such conditions is therefore a pressing  challenge. The existing work on implementing resilient consensus  mainly falls into  two categories: detector-based approaches and mean-subsequence-reduced (MSR) algorithms. The first one  originates  from  the diagnosis mechanism, {which requires each agent to be equipped with a detector in order to locate and isolate  malicious agents.} Its  essential idea  is to utilize the interaction outcoming   among  neighboring agents. Representative schemes include  reputation-based detector  \cite{zegers2021event}   to expose Byzantine agents, consensus-driven filters to discard compromised data \cite{mustafa2020resilient},  and two-hop information protocols {to}  suppress the intrusion of attacks and restore synchronization \cite{yuan2021secure, luo2023secure}.   In MSR algorithms, every benign agent discards extreme values from its neighboring agents  before { the state is updated, under the assumption that the number of adversaries does not exceed a known bound} \cite {leblanc2013resilient,ishii2022overview,8795564}.  Specifically, each normal agent  {removes}     all {potential}  outliers  {in accordance with the network's robustness constraints\cite{usevitch2020determining}.}  Furthermore, MSR algorithms have  been extended to resilient convex-optimization problems via integer programming \cite{sundaram2018distributed}. The problem of resilient formation control for multiple robots has also been investigated in \cite{10354416}.
It is worth emphasizing that the  MSR algorithms  are  convenient for practical operation   and can be fully distributed.

Both paradigms, however, hinge on abundant communication among normal agents. For the detector-based defense approaches, most of them require that, after isolating malicious agents, the remaining benign agents stay connected within the original communication graph \cite{luo2023secure,ZHAO2023110934}. {This progress simultaneously creates  severe vulnerabilities  because an attacker with access to the global communication topology can deliberately target critical agents, thereby disrupting connectivity.}  
For MSR algorithms, the concept of graph robustness has been  presented  to enhance graph resilience  for the purpose of avoiding   the aforementioned drawbacks   \cite{leblanc2013resilient,ishii2022overview,an2024mean}. It should be pointed out that  the lack of detectors can easily cause  false isolation of normal ones, resulting in unnecessary losses and default of critical  information. At the same time, it also increases  the occupation  of communication resources due to the edge-redundancy based defense framework. Actually, given {the potentially high cost} of communications in various applications, it is {crucial} to investigate how to achieve resilience while minimizing communication\cite{pirani2023graph}. In other words, this work provides some guidelines  for maintaining  secure operation under low communication resources.

Motivated by the {above}  discussions, two {persisting limitations are recognized}: (i) the stringent network connectivity requirements imposed on normal agents, and (ii) the excessive communication overhead inherent in redundancy-based defenses.  To  {address}  both challenges,  we propose  an active secure neighbor selection (ASNS) strategy to achieve resilient consensus for MASs subject to  Byzantine attacks. {
While this work is partially inspired by \cite{shao2023distributed}, it is worth noting that the method in \cite{shao2023distributed} neither considers security issues nor provides defense mechanisms. 
Therefore, an active neighbor selection mechanism under adversarial conditions should be considered.} 
The crux of the challenge lies in achieving network  connectivity among normal agents   through local neighbor selection rules   while eliminating the influence of attacks. Because connectivity quantifies how components stay interoperable, it {is normally  measured using metrics including  
 vertex/edge connectivity\cite{DION20031125} and  graph robustness \cite{usevitch2020determining}, etc.}  
 The  contributions of this paper   can be summarized as follows:

\begin{enumerate}
\item By exploiting a pre-discriminative graph, the proposed active secure neighbor selection (ASNS) strategy guarantees resilient consensus by actively forming a directed spanning tree whenever the attack changes. Compared with \cite{luo2023secure} and  \cite{ZHAO2023110934}, the proposed strategy removes the requirement for the persistent-connectivity assumption that was needed for normal agents.
\item The ASNS strategy allows for a flexible number of selected neighbors in the communication graph,  which provides more possibilities to improve the performance of MASs over the MSR algorithms \cite{leblanc2013resilient,sundaram2018distributed}. Based on this framework, minimum resilient communication with reduced overhead is further achieved through the selection of one in-neighbor.

\item The effectiveness of our work is validated through examples
involving dynamic Byzantine attacks.  On  low-robustness communication graphs, the ASNS strategy exhibits stronger resilience than MSR algorithms \cite{leblanc2013resilient,sundaram2018distributed}. When network connectivity among normal agents is disrupted, it also achieves better recoverability than \cite{luo2023secure,ZHAO2023110934} by reconstructing the topology.
\end{enumerate}

The remainder of this paper  is organized as follows. Section
$\rm\ref{GOsec2}$ reviews some notations  and graph-theoretic preliminaries. The system description and attack model are formulated  in Section 
$\rm\ref{GOsec3}$. The ASNS strategy {along}   with its analytical guarantee is presented in {Section}  $\rm\ref{GOsec4}$. {Section $\rm\ref{GOsec5}$  provides numerical experiments that demonstrate the effectiveness of the proposed methodologies,}  and {Section $\rm\ref{GOsec6}$ concludes the paper with some remarks on future research.}

\section{Preliminaries}\label{GOsec2}
Define    $\mathbb{R}^{n}$   as the set of  $n$-dimensional real vectors and  $\mathbb{R}^{n\times m}$ the set of  $n\times m$-dimensional real matrices. $\mathbb{Z}^{+}$ denotes the set of positive integers.
$\bm 1$ and  $\bm I$ represent  a column   vector whose entries are all 1 and  an identity matrix   with appropriate dimensions.   ${\rm diag}\left\{x\right\}$ stands for a   diagonal matrix with diagonal entries  being the elements of vector  $x$.  $\left| \cdot \right|$ represents 
the cardinality of a set. For some positive  integer  $r$,  let   $\underline{r}\triangleq \left\{ 0,\dots ,r \right\}$.   Moreover, the  $i$-th element of vector $x$ is written  as $x_{(i)} \in \mathbb{R}$.

Let ${\mathcal{G}(k)}=({\mathbb{E}(k)},\mathbb{V})$ be a directed graph with  $N$ nodes, where $\mathbb{E}(k)$ is the  edge set and $\mathbb{V}$ is the  node  set. A directed edge $(j,i)\in \mathbb{V}$ signifies  an ordered edge connection from $v_{i}$ to $v_{j}$, where nodes $v_i$ and $v_j$ are called  parent  node and child  node respectively. If  a node has  ordered paths to  preserve  all  other nodes in the graph, it is called   the rooted node.  {$N_i^{+}(k)$} and {$N_{i}^{-}(k)$} are    sets of  in-neighbors and out-neighbors for  agent $i$ at time $k$.  {$A(k)=[a_{ij}(k)]\in \mathbb{R}^{N\times N}$} is  a weighted
adjacency matrix: $a_{ij}(k)\neq 0$ if {$j\in N_i^{+}(k)$} and $a_{ij}(k)=0$ otherwise for $j\neq i$. Moreover,  $a_{ii}(k)= 0$. Define the  Laplacian matrix of {$\mathcal{G}(k)$}
as {$L(k)=[l_{ij}(k)]\in \mathbb{R}^{N\times N}$} in which   $l_{ij}(k) =-a_{ij}(k)$ $(i\neq j)$ and {$l_{ii}(k) =\sum_{j\in N_i^{+}(k)}{a_{ij}(k) \,\,}$}.  Table $\rm{\ref{al23}}$ summarizes some other important notations.

\begin{table}[H]
	\caption{Notations}
	\label{al23}
	\renewcommand\arraystretch{1.6}
	\setlength{\tabcolsep}{3pt}
	\begin{tabular}{p{40pt} p{200pt}}
		\hline
		Symbol& Definition\\
		\hline
		$\bar{\mathcal{B}}$ & The set of attack-admissible agents\\
        $\bar{\mathcal{A}}$ & The normal agent set \\
		${\mathcal{B}}(0,k)$ & The set of Byzantine agents during $[0,k]$ \\
		${\mathcal{A}}(0,k)$ & The set of normal agents staying during $[0,k]$, i.e.,$\mathbb{V}\backslash {\mathcal{B}}(0,k)$ \\
		$\mathcal{G}(k)$ & The graph of agents in $\mathbb{V}$ at time $k$ \\
		$\mathcal{G}_{\mathcal{A}}(k)$ & The subgraph of agents in ${\mathcal{A}}(0,k)$ corresponding  to $\mathcal{G}(k)$ at time $k$ \\
		$\mathcal{G}_{pre}(k)$ & The pre-discriminative graph at time $k$ \\	
		$\mathcal{G}_{\mathcal{A}\text{-}pre}(k)$ & The subgraph of agents in ${\mathcal{A}}(0,k)$ corresponding  to $\mathcal{G}_{pre}(k)$ at time $k$ \\	      
		$L_{\mathcal{A}\text{-}pre}(k)$ & The corresponding Laplacian  matrix of $\mathcal{G}_{\mathcal{A}\text{-}pre}(k)$  at time $k$ \\
		$\Omega(k)$ & The   state  set  of agents in $\mathcal{A}(0,k)$\\
       $\Xi(k)$& The convex hull formed by the states   in $\Omega(k)$\\
		\hline
	\end{tabular}
\end{table}

\section{Problem Formulation}\label{GOsec3}
\subsection{System Model}\label{GOsec223}
{For an MAS, an attack-free agent $i$ can be modeled as\cite{ren2008distributed}:}

\begin{equation}\label{GOeq499}
	\begin{aligned}
		x_i\left( k+1 \right) =x_i\left( k \right) +\epsilon \sum_{j\in N_i^{+}\left( k \right)}{a_{ij}\left( k \right) \left( x_j\left( k \right) -x_i\left( k \right) \right),}
	\end{aligned}
\end{equation}
where $x_i\left( k \right) \in \mathbb{R}^n$ is the state,  $\epsilon \in \mathbb{R}$ is the  step-size  with $\epsilon \in \left(0,\mathop {\frac{1}{\sum_{j\in N_i^{+}}{a_{ij}(k) \,\,}}}\right)$.

For convenience, the  system $(\ref{GOeq499})$ is rewritten  as 
\begin{equation}\label{GOeq4919}
	\begin{aligned}
		x_i\left( k+1 \right) =x_i\left( k \right) -\epsilon \sum_{j\in \widetilde{N_i}\left( k \right)}{l_{ij}\left( k \right) x_j\left( k \right),}
	\end{aligned}
\end{equation}
where $\widetilde{N_i}\left( k \right) =N_i^{+}\left( k \right) \cup \left\{ i \right\}$. Thus its  augmented  form becomes  $x\left( k+1 \right) =\left( \bm I -\epsilon L\left( k \right)\otimes \bm I \right) x\left( k \right)$ where  $x(k)=[x_1^{\top}(k), x_2^{\top}(k),\dots,x_{N}^{\top}(k)]^{\top}$.

In what follows, the concept of  robustness   is  provided  to characterize  the connectivity  of a network.

\begin{definition}\label{GOde10k}
	\emph{($r$-reachable and $r$-robustness \cite{ishii2022overview})} In $\mathcal{G}=(\mathbb{E},\mathbb{V})$, given $r\in \mathbb{Z}^{+}$, a
	nonempty set $Q_{0}\subset  \mathbb{V}$ is  said to be  $r$-reachable, if there exists $i\in Q_{0}$ such that  $|N_i^{+}\backslash Q_0|\geq r$ where $N_i^{+}$ is the set of in-neighbors  of agent $i$. For any two nonempty disjoint   subsets $Q_{a}, Q_{b}\subset  \mathbb{V}$, $\mathcal{G}$ is $r$-robust if either of them is $r$-reachable. 
\end{definition}

\subsection{Attack Model}
  This paper focuses  on   the Byzantine attacks \cite{pirani2023graph},  which is a kind of   flexible attack strategies  on the  agent layer. It is capable of  transmitting  different values  to different neighbors at each time $k$. Here the   normal agent set is $\bar{\mathcal{A}}$  and  the set of attack-admissible agents is  $\bar{\mathcal{B}}$, that is,    $\bar{\mathcal{A} }\cup \bar{\mathcal{B} }=\mathbb{V}$ {and}  $\bar{\mathcal{A} }\cap \bar{\mathcal{B} }=\varnothing$. Let $\mathcal{B}_{i}(k)$  be    the set of Byzantine agents in $N_i^{+}\left( k \right)$ at time $k$ for agent $i$.  Moreover,   $\mathcal{B} (0,k)=\bigcup_{l\in \underline{k}}{\big( \bigcup_{i\in\mathbb{V}}{\mathcal{B}_{i\,\,}}(l)\big)}$   represents  the set of  Byzantine agents  from initial time $0$ to time $k$ and $\mathcal{A}(0,k)=\mathbb{V}\setminus\mathcal{B}(0,k)$. It is clear that  $\overline{\mathcal{A}}\subseteq \mathcal{A}(0,k)$. Let $\Omega(k)$ be  the   state set   of agents in $\mathcal{A}(0,k)$.  Define $\Xi(k)\triangleq\operatorname{Conv}(\Omega(k))$ as    the convex hull formed by the states   in $\Omega(k)$.

If agent $i$ is a Byzantine agent  at time $k$, then  \begin{equation}\label{GOeq4117}
	\begin{aligned}
		x^{a}_{ij}\left( k \right) =f_{ij}\left( k \right),~j\in 
		N_i^{-}(k),
	\end{aligned}
\end{equation}
where  $x^{a}_{ij}\left( k \right)  \in \mathbb{R}^n $ is the state  transmitted  from agent $i$ to its neighbor $j$ and  $f_{ij}\left( k \right)\in \mathbb{R}^n$ is the attack signal.

The following assumption is essential to the developments in this paper.

\begin{assumption}\label{GOas02} \hspace{-0.001cm} \cite{ishii2022overview,yuan2024resilient}
	($F$-local attack model) For each agent, there  are at most $F$ Byzantine  agents in its  in-neighbors.  The system cannot    be  attacked at the initial time.
\end{assumption} 

\begin{remark}
	The $F$-local attack  model includes  the $F$-total   strategy  which limits  the number of Byzantine agents  on a global scale to $F$. Besides, the $F$-local  model  is  more  suitable   for the  situation where the number  of misbehavior agents  varies with network size and  connectivity \cite{ishii2022overview}. {Actually,  such attacks pose a more severe threat.}
\end{remark}

To {identify}  potential anomalies,  we employ an attack detector that leverages  two-hop information \cite{YUAN2025111908,yuan2021secure,10102299}. Specifically, at every time $k\geqslant1$,  each   agent $i \in \mathbb{V}$ transmits the packet  $\left\{ x_i\left( k \right) ,\left\{ j,x_j\left( k-1 \right) \right\} _{j\in N_{i}^{+}(k-1)} \right\}$  to  its out-neighbors. During  the detection process,  for  each  normal agent $i \in \mathcal{A} (0,k)$, the detection strategy with respect to   agent $j\in {N_i^{+}}\left( k \right)$ admits  
\begin{equation}\label{GOeqw117}
\begin{alignedat}{1}
\left\{
  \begin{aligned}
    x_{j}(k) &\ne x_{j}(k-1) + \sum l_{jh}(k-1)x_h(k-1),~ j \in \mathcal{B}_{i}(k),\\
    x_{j}(k) &=   x_{j}(k-1) + \sum l_{jh}(k-1)x_h(k-1),~ j \notin \mathcal{B}_{i}(k).
  \end{aligned}
\right.
\end{alignedat}
\end{equation}
This control protocol-based detection approach  is partially  inspired by  \cite{yuan2021secure,ishii2022overview}.

\begin{definition}\label{GOde2}
		\emph{(Resilient Consensus) \cite{yan2022resilient}} For the Byzantine  attacks,  a multi-agent system is said to realize  resilient consensus if $\lim_{k\to\infty}\lVert{x}_{i}(k)-{x}_{j}(k)\rVert=0,~  \forall~i,j \in \bar{\mathcal{A}}$.
\end{definition}

The  {objective}  of this paper is to {develop}  an active secure neighbor selection strategy that ensures  resilient consensus while relaxing the restrictions on graph connection among normal agents with  {low   communication overhead.}

\section{Main Results}\label{GOsec4}
In this part, we  will propose  a defense framework for active secure neighbor selection. More specific, it consists of   two steps: 1) construction of pre-discriminative graph, and  2) design of active secure neighbor selection strategy.  These tasks will be addressed one by one.

\subsection{Construction of Pre-discriminative Graph}\label{GOsubsec112}
In this subsection,   a pre-discriminative graph is  constructed  to pave {the}  way for the secure  neighbor selection. To this end, we first introduce the concept of  pre-discriminative graph for all agents in $\mathbb{V}$. This graph specifies the range of neighbors that an agent can select from the normal ones.

\begin{definition}\label{GOde2}
	\emph{(Pre-discriminative Graph)} 
	For a given integer $k$,   the pre-discriminative graph is defined as  
$\mathcal{G}_{pre}(k)\triangleq(\mathbb{E}_{pre}(k),\mathbb{V})$, where  
the  set $\mathbb{E}_{pre}(k)$ comprises every edge through which an agent chooses normal neighbors. For each agent $i\in\mathbb{V}$, its  neighbor set in $\mathcal{G}_{pre}(k)$, {termed as}  the candidate neighbor set, is defined as 
$N_{i\text{-}pre}(k)\triangleq\{\, j\in \mathbb{V}\mid (j,i)\in \mathbb{E}_{pre}(k)\,\}$.
\end{definition}

\begin{figure*}[htbp]
	\centering
	\includegraphics[width=0.85\linewidth]{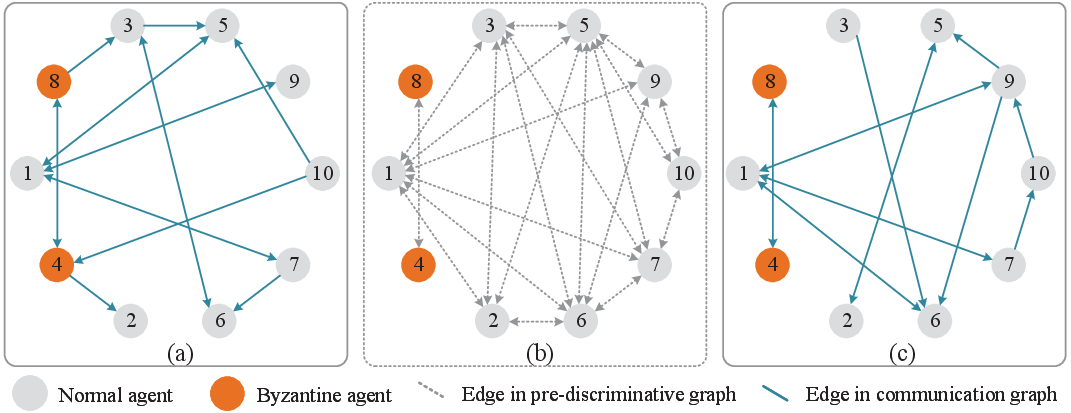}
	\caption{(a) $\mathcal{G}(k-1)$: the 
		communication   graph corresponding to time $k-1$ under attacks occurring at time $k$; (b)   $\mathcal{G}_{pre}(k)$:  the 
		pre-discriminative  graph at  time  $k$ ; (c)  $\mathcal{G}(k)$: the 
		communication   graph at  time  $k$ after the ASNS strategy with $\mathcal{G}(k)\subseteq \mathcal{G}_{pre}(k)$.}
	\label{GOfig9qqqq90}
\end{figure*}

\begin{remark}
	It is noted that  the actual communication topology  is not $\mathcal{G} 
	_{pre}(k)$;  rather,  it is a subgraph of $\mathcal{G} 
	_{pre}(k)$, i.e., $\mathcal{G} 
	(k)\subseteq\mathcal{G} 
	_{pre}(k)$,  as  illustrated  in Figs. $\ref{GOfig9qqqq90}$(b)-(c). Actually, the neighbor information associated with each agent in $\mathcal{G}_{pre}(k)$  reflects the reorganized range of available neighbors after the isolation strategy, providing reliable candidate agents for subsequent neighbor selection for $\mathcal{G}(k)$. 
	Consequently, the selected neighbor set satisfies $N_i^{+}(k) \subseteq N_{i\text{-}pre}(k)$. Fig. \ref{GOfig9qqqq90} depicts a ten-agent system. When  agents 4 and 8 are under attacks (see  Fig.~\ref{GOfig9qqqq90}(a)), the pre-discriminative graph is first  constructed as illustrated in Fig.~\ref{GOfig9qqqq90}(b). Then the communication graph is reconstructed in terms of the ASNS strategy  as shown in Fig.~\ref{GOfig9qqqq90}(c), confirming   $\mathcal{G}(k)\subseteq \mathcal{G}_{pre}(k)$.
\end{remark}

Next, the pre-discriminative graph $\mathcal{G}_{pre}(k)$ is constructed. Specifically,  the information of  $\mathcal{B}_{i}(k)$ is first  broadcasted to eliminate any possibility of establishing links between normal and compromised agents. Then the actual reconstruction of $\mathcal{G}_{pre}(k) $ is triggered only when new  Byzantine agents {are detected}, i.e.,   $\mathcal{A}(0,k) \neq \mathcal{A}(0,k-1)$ and set   $k=k_s$ where $s\in \mathbb{Z}^{+}$. This avoids frequent invocation of the subsequent updates to the pre-discriminative  graph and communication graph, thereby reducing defense overhead.   In particular,	each agent $i \in \mathcal{A}(0,k)$ rebuilds   undirected edges with the agents belonging to $N_{i\text{-}  pre}(0)\cap\mathcal{A}(0,k)$  for  $\mathcal{G}_{pre}(k)$. In this way, the attacked agents will be isolated from the normal ones to ensure the secure candidate neighbor range. \textbf{\emph{Algorithm} \ref{al01e}} summarizes  the specific steps.

Through the above construction process, it  can be  seen  that  $\mathcal{G}_{pre}(k)$ is   undirected. Let   $\mathcal{G}_{\mathcal{A}\text{-}pre}(k)$ denote  the  subgraph   induced by the agent set   ${\mathcal{A}}(0,k)$  within  	$\mathcal{G}_{pre}(k)$ at time $k$, as illustrated in  Fig. $\ref{GOfig23412w0}$(a).
Next, 
the network connectivity  of $\mathcal{G}_{\mathcal{A}\text{-}pre}(k)$, will be analyzed  to   pave a way  to  the connection performance  preservation  among normal agents  of the communication graph  $\mathcal{G}_{\mathcal{A}}(k)$ (see Fig. $\ref{GOfig23412w0}$(b))   after the ASNS strategy.

\begin{figure}[H]
	\vspace{-0.2cm}
	\centering
	\includegraphics[width=\linewidth]{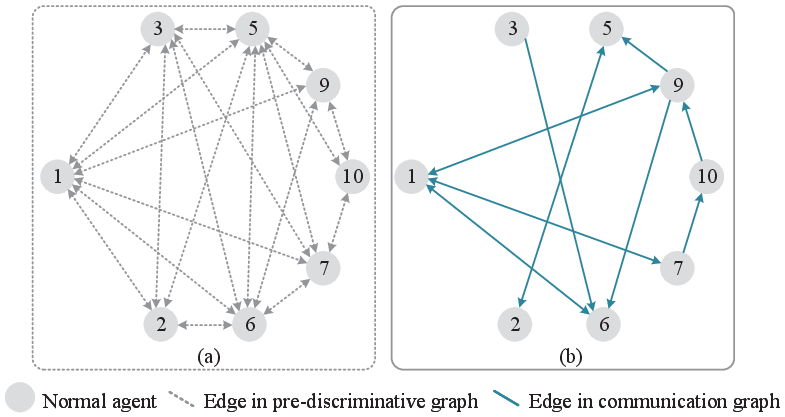}
	\caption{  (a) $\mathcal{G}_{\mathcal{A}\text{-}pre}(k)$: the subgraph  of agents in ${\mathcal{A}}(0,k)$ corresponding to $\mathcal{G}_{pre}(k)$ in Fig. $\ref{GOfig9qqqq90}$(b); (b) $\mathcal{G}_{\mathcal{A}}(k)$: the subgraph  of agents in ${\mathcal{A}}(0,k)$ corresponding to $\mathcal{G}(k)$ in Fig. $\ref{GOfig9qqqq90}$(c) with $\mathcal{G} _{\mathcal{A}}(k)\subseteq \mathcal{G} _{\mathcal{A}\text{-}pre}(k)$.}
	\label{GOfig23412w0}
\end{figure}

\begin{algorithm}
	\caption{ Pre-Discriminative Graph $\mathcal{G}_{pre}(k)$ Construction Strategy}  
	\label{al01e} 
	\begin{algorithmic}[1] 
		\For{$k>0$} 
		\For{each normal  agent $i \in \mathcal{A} (0,k-1)$}   
		\State  \hspace{-12mm}  \textbf{\textit{\underline{(Attack detection)}}} 
		\For{each $j \in N_i^{+}(k-1) \cap\mathcal{A} (0,k-1)$}   
		\State Implement the detection strategy  $(\ref{GOeqw117})$;
		\EndFor     
		\State  \hspace{-12mm} \textbf{\textit{\underline{(Broadcast)}}}
		\State  $\mathcal{B}_{i}(k)$ is broadcasted at time $k$; 	
		\If {$\mathcal{A}(0,k) \neq \mathcal{A}(0,k-1)$} 
		\State \hspace{-17mm}\textbf{\textit{\underline{(Graph construction)}}}         
		\State  Construct  the  pre-discriminative graph $\mathcal{G}_{pre}(k)$: 	each agent $i \in \mathcal{A}(0,k)$ rebuilds   undirected edges with  agents belonging to  $N_{i\text{-}pre}(k-1)\cap\mathcal{A}(0,k)$ for  $\mathcal{G}_{pre}(k)$.	
		\EndIf
		\EndFor 
		\EndFor  
	\end{algorithmic}  
\end{algorithm}

\begin{proposition}\label{GOth1}
	For MASs suffering from Byzantine attacks, under \textbf{\emph{Algorithm} \ref{al01e}}, $\mathcal{G}_{\mathcal{A}\text{-}pre}(k)$,  the subgraph  of $\mathcal{G}_{pre}(k)$ among all agents in  ${\mathcal{A}}(0,k)$ is connected if  the initial pre-discriminative graph $\mathcal{G}_{pre}(0)$ is $(F+1)$-robust. 
\end{proposition}

\begin{proof}
     The worst-case attack  scenario within  the interval $[0,k]$ is considered. For each normal agent $i$, if $F < \lvert N_{i\text{-}pre}(0) \rvert$, it has exactly $F$ Byzantine candidate neighbors; otherwise, all candidate neighbors are compromised. This condition is formally expressed as 
\[
\bigg| \bigg( \bigcup_{l \in [0,k]} N_{i\text{-}pre}(l) \bigg) \cap \mathcal{B}(0,k) \bigg|
= \min \left\{ F, \lvert N_{i\text{-}pre}(0) \rvert \right\}.
\]

First, the preliminary form of the isolation process  in step 11 of \textbf{\emph{Algorithm} \ref{al01e}} is considered, where all directed edges from Byzantine agents to normal agents are removed by the defense strategy. Since $\mathcal{G}_{pre}(0)$ is $(F+1)$-robust,  under the above attack scenario and defense scheme, it follows  that $\mathcal{G}_{pre}(k)$ is $1$-robust. Notably, using  this isolation mechanism, all communication edges between $\mathcal{B}(0,k)$ and $\mathcal{A}(0,k)$ are directed from $\mathcal{A}(0,k)$ to $\mathcal{B}(0,k)$. Consequently, $\mathcal{G}_{\mathcal{A}\text{-}pre}(k)$ is at least $1$-robust in this setting.

Next, consider the actual  isolation mechanism in the ASNS strategy, where all undirected edges between Byzantine and normal agents are removed. Consequently,  it follows directly that $\mathcal{G}_{\mathcal{A}\text{-}pre}(k)$ is also  at least $1$-robust. Hence, $\mathcal{G}_{\mathcal{A}\text{-}pre}(k)$ contains a directed spanning tree. Since all edges among ${\mathcal{A}}(0,k)$ in $\mathcal{G}_{\mathcal{A}\text{-}pre}(k)$ are undirected, the graph is connected. Thus the proof is complete. 
\end{proof}

Subsequently, based on the pre-discriminative  graph constructed above, we perform a  preprocessing step on  system~(\ref{GOeq499}) so that the forthcoming ASNS strategy can rely on well-defined selection criteria. 

Recent research \cite{shao2023distributed} has  shown that, in semi-autonomous networks,  connectivity under neighbor selection  can be  determined by the normalized eigenvector linked to the smallest eigenvalue. This eigenvalue arises from a perturbed Laplacian matrix formed by  the original Laplacian matrix  and the input matrix.

To exploit this,  system (\ref{GOeq499}) is preprocessed to emulate a class of semi-autonomous ones. Specifically, choose  any agent   in  ${\mathcal{A}}(0,k)$ as a virtual leader with no influence caused by  external input   and  the model  in $(\ref{GOeq499})$ can be    transformed as     
\begin{equation}\label{GOeq487}
	\begin{aligned}
		x_i\left( k+1 \right) =&~x_i\left( k \right) +\epsilon\sum_{j\in N_i^{+}\left( k \right)}{a_{ij}\left( k \right) \left( x_j\left( k \right) -x_i\left( k \right) \right)}\\&-\sum_{p=1}^mb_{ip}({x}_i(k)-{u}_p(k)),
	\end{aligned}
\end{equation}
where $u_{p}(k)\in \mathbb{R}^{n}$ is the $p$-th virtual external input with ${u}_p(k)=x_i\left( k \right) (b_{ip}=1)$ in order to  offset the impact of virtual input on   system $(\ref{GOeq499})$. Here,   $b_{ip} \in  \mathbb{R}$ is the weight  coefficient of input: $b_{ip} = 1$ if agent ${i}$ is  designed as a virtual leader injected by $u_{p}(k)$ and   $b_{ip} = 0$ otherwise.

Thus, the augmented  dynamics of  agents in ${\mathcal{A}}(0,k)$ admits  
$$x\left( k+1 \right) =-\left(\epsilon L_B\left( k \right) \otimes \bm I \right) x\left( k \right) +\left(\epsilon B\otimes \bm I \right) u\left( k \right),
$$
where  $x(k)=[x_1^{\top}(k), x_2^{\top}(k),\dots,x_{\left| \mathcal{A}(0,k) \right|}^{\top}(k)]^{\top}$, $u(k)=[u_1^{\top}(k), u_2^{\top}(k),\dots,u_{m}^{\top}(k)]^{\top}$ and $L_B(k)=L_{\mathcal{A}\text{-}pre}\left( k \right)+\mathbf{diag}(B\cdot\bm 1)$ is a perturbed Laplacian matrix  with   $B=(b_{ip})\in\mathbb{R}^{m\left| {\mathcal{A}}(0,k) \right|}$ and $L_{\mathcal{A}\text{-}pre}(k)$ is   the corresponding Laplacian  matrix 
of $\mathcal{G}_{\mathcal{A}\text{-}pre}(k)$.

\begin{lemma}\label{GOpro29}
	If  $\mathcal{G}_{pre}(0)$ is $(F+1)$-robust, then the smallest eigenvalue 
	$\lambda _1\left( L_B\left( k \right) \right)>0$ is a simple eigenvalue of $ L_B\left( k \right) $ and its associated eigenvector $v _1\left( L_B\left( k \right) \right)$  can be chosen strictly  positive.
\end{lemma}

\begin{proof}
	By   \textit{Proposition}   $\ref{GOth1}$,  $\mathcal{G}_{\mathcal{A}\text{-}pre}(k)$  is connected.   The statement then follows immediately from  \textit{Lemma} 1 in \cite{shao2023distributed}. 
\end{proof}

In~\cite{shao2023distributed}, for the original neighbor set   $N^{+}_{i}(k)$, each agent selects in-neighbors satisfying $v_{1(j)}(L_B(k)) < v_{1(i)}(L_B(k))$. This strategy prunes redundant edges, and accelerates system convergence while preserving network connectivity. However, it offers no protection against  adversarial attacks, whose misreported states can render the criterion insecure and destabilize the system. Therefore, the subsequent investigation centers on an active neighbor selection framework  against  attacked  agents.

\subsection{Design of Active Secure Neighbor Selection Strategy}\label{GOsubsec3}
\vspace{0.15cm}
Here  an   ASNS strategy  is designed  to reconstruct  $\mathcal{G} \left( k \right) $ and ensure  the resilient consensus. 

The ASNS strategy begins with the pre-discriminative graph reconstruction executed by \textbf{\emph{Algorithm} \ref{al01e}}, during which attacked agents  are exposed  and the detection results  are broadcasted. Subsequently, normal agents actively establish communication links by selecting secure neighbors, rather than passively removing untrusted ones. The detailed procedure of the ASNS strategy is presented in \textbf{\emph{Algorithm} \ref{al01}}. Specifically, at each time $k$, the main process implemented by each normal agent $i$ in $\mathcal{A}(0,k-1)$ is described  as follows.

\textbf{Attack detection and broadcast ({Step} 6)} and \textbf{Pre-discriminative graph construction  ({Step} 9):} The detailed procedure has been provided in \textbf{\emph{Algorithm}~\ref{al01e}}.

\textbf{Active secure neighbor selection  ({Steps} 11-19):}  
Based on  the pre-discriminative graph $\mathcal{G}_{pre}(k)$,   a  virtual leader $\tilde{i}$ in $\mathcal{A}(0,k)$ is  first  selected  and the perturbed Laplacian is constructed  as $L_B(k)\triangleq  L_{\mathcal{A}\text{-}pre}\left( k \right)+\mathbf{diag}(B\bm 1)$ in the foundation of (\ref{GOeq487})   with $u_p(k)=x_{\tilde{i}}\left( k \right)$. Define  $\psi_i(k)$ as  the set of  agents  in $N_{i\text{-}pre}(0)\cap\mathcal{A}(0,k) \backslash\{\tilde{i}\}$ where each agent $j \in \psi_i(k)$  satisfies  $v _{1(i)}\left( L_B\left( k \right) \right)> v_{1(j)}\left( L_B\left( k \right) \right)$.   Then each agent selects    the set of in-neighbors   satisfying $N^{+}_{i}(k) \subseteq \psi_i(k)$ and $\left| N^{+}_{i}(k) \right| \neq 0$.    

\begin{algorithm}
	\caption{Active Secure Neighbor  Selection  Strategy for Flexible Communication} 
	\label{al01} 
	\begin{algorithmic}[1] 
		\Require $\mathcal{G}(0)$ and $F$.%
		\Ensure $\mathcal{G}(k)$. \\    
		\textbf{Initialization:} Each agent $i\in \mathbb{V}$ initializes its information set $N_i^{+}(0)$ and $x_i(0)$; $\mathcal{A} (0,0)=\mathbb{V}$;  $\mathcal{B}_{i} (0)=\varnothing$;\\
		\textbf{Iteration:} 
		\For{$k>0$} 
		\For{each normal agent $i \in \mathcal{A} (0,k-1)$}
		\State \hspace{-11mm}\textbf{\textit{\underline{(Attack detection and broadcast)}}}   
		\State \textbf{Steps 4-8} in \textbf{\emph{Algorithm}  \ref{al01e}};
		\If {$\mathcal{A}(0,k) \neq \mathcal{A}(0,k-1)$}
		\State \hspace{-16mm}\textbf{\textit{\underline{(Pre-discriminative graph  construction)}}}  
		\State \textbf{Step 11} in \textbf{\emph{Algorithm}  \ref{al01e}};
		\State \hspace{-16mm}\textbf{\textit{\underline{(Active secure neighbor selection)}}}
		\State Set an agent $\widetilde{i}$ in $\mathcal{A}(0,k)$ as the virtual leader;
		\State Calculate $L_{{B}}(k)$;
		\For{each $j \in N_{i\text{-}pre}(0)\cap\mathcal{A}(0,k) \backslash\{\tilde{i}\}$}   
		\If{$v _{1(i)}\left( L_B\left( k \right) \right)>v_{1(j)}\left( L_B\left( k \right) \right)$}  
		\State Classify  agent $j$ into  $\psi_i(k)$ which is the pre-discriminative neighbor selection set of agent $i$;
		\EndIf  
		\EndFor
		\State Set $N_i^{+}(k)=\varnothing$;
		\State  Construct  the  communication  graph  $\mathcal{G}(k)$:  choose   the set of in-neighbors   satisfying $N_{i}^{+}(k) \subseteq \psi_i(k)$ and $\left| N^{+}_{i}(k) \right| \neq 0$.	
		\EndIf
		\EndFor 
		\EndFor  
	\end{algorithmic}  
\end{algorithm}

The performance of the above ASNS strategy  is examined next. We first analyze the network connectivity  of  $\mathcal{G}_{\mathcal{A}}(k)$ which is defined as the  subgraph  induced by the  agents in ${\mathcal{A}}(0,k)$ corresponding  to 	$\mathcal{G}(k)$ at time $k$,  as depicted in  Fig. $\ref{GOfig23412w0}$(b).
The issue  of convergence will be investigated in terms of  the above graph connection performance  analysis.

Now we   discuss  the feasibility  of \textbf{\emph{Algorithm}  \ref{al01}} by deriving the conditions that guarantee every agent except the rooted one  in $\mathcal{A} (0,k)$ can always find at least one admissible in-neighbour.

\begin{proposition}\label{GOth190}
	If $\mathcal{G}_{pre}(0)$ is $(F+1)$-robust, then    $\psi_i(k) \neq \varnothing$ for each  agent $i \in \mathcal{A}(0,k) \backslash\{\tilde{i}\}$.                                 
\end{proposition}

\begin{proof}
	Based on the transformed system $(\ref{GOeq487})$, we  proceed   by contradiction. For simplicity, we omit time $k$   hereafter.
	
	Suppose that  there exists  an agent $i \in \mathcal{A}(0,k) \backslash\{\tilde{i}\}$, such that $v_{1(i)}<v_{1(j)}$, for all $j\in N_{i\text{-}pre}(0)\cap\mathcal{A} (0,k)$.  For the $i$-th row of $L_B v_1=\lambda_1(L_B)v_1$, we have 
    	\begin{equation}\label{GOeq4871}
		\begin{aligned}
			\bigg(\sum_{j\in\varXi(0,k) }l_{ij}\bigg)v_{1(i)}-\sum_{j\in\varXi(0,k)}l_{ij}v_{1(j)}=\lambda_1(L_B)v_{1(i)}.
		\end{aligned}
	\end{equation}
    where $\varXi (0,k)\triangleq N_{i\text{-}pre}(0)\cap \mathcal{A} (0,k)$.
	
	If  $v_{1(i)}<v_{1(j)}$ for all $j\in N_{i\text{-}pre}(0)\cap\mathcal{A} (0,k)$, one gets $\lambda_1(L_B)v_{1(i)}<0$, which is  a contradiction with \textit{Lemma} $\ref{GOpro29}$. Thus, $\psi_i(k) \neq \varnothing$ for each  agent $i \in \mathcal{A}(0,k) \backslash\{\tilde{i}\}$.
\end{proof}

The network connectivity among   normal agents $\mathcal{G}_{\mathcal{A}}(k)$ is now guaranteed. By leverage of \cite{ishii2022overview},   it is indicated that  under  the ASNS strategy and  the condition that  $\mathcal{G}_{pre}(0)$ is $(F+1)$-robust, $\mathcal{G}_{\mathcal{A}}(k)$ is ensured  to contain  a spanning tree for all $k$.

Under  the ASNS strategy,  there exists no edge between agents in 	$\mathcal{B}(0,k)$ and $\mathcal{A} (0,k)$. In other words, the agents in $\mathcal{A}(0,k)$ will not be affected by the attackers. Thus, the state evolution of agents in $\mathcal{A} (0,k)$ is governed by 
\begin{equation}\label{GOwq127}
	\begin{aligned}
		x\left( k+1 \right) =\left( \bm I -\epsilon L_{\mathcal{A}\text{-}pre}\left( k \right)\otimes \bm I \right) x\left( k \right),
	\end{aligned}
\end{equation}
where $x\left( k \right) \in \mathbb{R} ^{\left| \mathcal{A} (0,k) \right|}$ and $\boldsymbol{I}\in {\mathbb{R} ^{\left| \mathcal{A} (0,k) \right|\times {\left| \mathcal{A} (0,k) \right|}}}
$.

The resilient-consensus property is formally established in the next theorem.

\begin{theorem}\label{GOth330}
	Consider the  MASs $(\ref{GOeq499})$ subject to   Byzantine attacks $(\ref{GOeq4117})$. Under  the ASNS strategy  and  Assumption      $\ref{GOas02}$, if $\mathcal{G}_{pre}(0)$ is $(F+1)$-robust, the  resilient consensus can be achieved by agents in $\bar{\mathcal{A}}$.
\end{theorem}

\begin{proof}
	Let 
	$\left\{ k_1,\dots ,k_{s}, k_{s+1}\dots \right\}$ be  the discrete time instants at which the attackers change their target set; i.e., $\mathcal{B} (0,k_s)\neq \mathcal{B} (0,k_s-1)$. At each time $k$ and  for  the $l$-th  dimension, we denote  the  maximum and  minimum state values  of agents in $\mathcal{A}(0,k)$ as    $x_{max(l)}(k)$ and $x_{min(l)}(k)$.  Let $P_{min(l)}(k)$ and $P_{max(l)}(k)$ be  the sets of agents in $\mathcal{A} (0,k)$  holding the state value as $x_{min(l)}(k)$ and $x_{max(l)}(k)$, respectively.   
	
	For convenience, rewrite $(\ref{GOeq4919})$ as  
	\begin{equation}\label{GOeqp71}
		\begin{aligned}
			x_i\left( k+1 \right) =(1-\epsilon l_{ii})x_i\left( k \right) -\epsilon \sum_{j\in {N_i^{+}}\left( k \right)}{l_{ij}\left( k \right) x_j\left( k \right)}.
		\end{aligned}
	\end{equation} 

    \vspace{0.1cm}
    
	During  interval $[k_{s\,\,},k_{s+1})$,   $\epsilon \in (0,\mathop {\frac{1}{\max l_{ii}}})$ ensures that all  coefficients  of $x_{i}(k)$ in $(\ref{GOeqp71})$ are  nonnegative and  sum to one. Hence, the state value of   each agent in $\mathcal{A} (0,k_s)$ is a convex combination of its own value and the values received   from its neighbors under  protocol $(\ref{GOeq499})$. Therefore, it has   $\Xi (k+1)\subseteq \Xi (k)$ for all  $k\in[k_{s\,\,},k_{s+1})$.  Besides, since there is no state jump occur at  instant $k_s$, we also have   $ \Xi (k_s^{+})=\Xi (k_s^{-})$. Then, the following outline of analysis is provided.

    \vspace{0.1cm}
    
    Since we have already established $\Xi(k+1) \subseteq \Xi(k)$ for the entire  process, 
    to verify resilient consensus, it remains to show that the time interval satisfying  $\Xi(k+1) = \Xi(k)$ is bounded. 
    To this end, since  it is obvious  that  $ \Xi (k_s^{+})=\Xi (k_s^{-})$,  the subsequent  proof proceeds with each interval $[k_s, k_{s+1})$ and is carried  at  each dimension of state $x_i\left( k\right)$ in $(\ref{GOeqp71})$. For the \( l \)-th dimension, we focus on the agents holding extreme values, i.e.,  \(i \in P_{\min(l)}(k) \cup P_{\max(l)}(k)\). Three  exhaustive cases are involved  at each time step \( k \):
	
    \vspace{0.1cm}
    
	\textbf{Case 1)}   $N_i^{+}(k)\cap 
	P_{min(l)}(k)=\varnothing,~\forall i\in P_{min(l)}(k)$ and $N_i^{+}(k)\cap 
	P_{max(l)}(k)=\varnothing,~\forall i\in P_{max(l)}(k)$;

    \vspace{0.1cm}
	
	\textbf{Case 2)} $N_i^{+}(k)\cap 
	P_{min(l)}(k)=\varnothing,~\forall i\in P_{min(l)}(k)$ and   $N_i^{+}(k)\cap 
	P_{max(l)}(k)\neq\varnothing,~\exists i\in P_{max(l)}(k)$; 
	
    \vspace{0.1cm}
    
	\textbf{Case 3)} $N_i^{+}(k)\cap 
	P_{min(l)}(k)\neq\varnothing,~\exists  i\in P_{min(l)}(k)$ and   $N_i^{+}(k)\cap 
	P_{max(l)}(k)\neq\varnothing,~\exists i\in P_{max(l)}(k)$.
    
    \vspace{0.1cm}
    
	Note that resilient consensus is achieved   if  $x_{\min(l)}(k) = x_{\max(l)}(k)$.  In what follows we consider the situation  that at least one dimension $l$ satisfies  $x_{min(l)}(k) \neq x_{max(l)}(k)$  before resilient consensus is  achieved.

    \vspace{0.1cm}
	
	\textbf{Case 1).} For every agent $i \in P_{min(l)}(k)\cup P_{max(l)}(k)$, the ASNS strategy guarantees an in-neighbor $i \in P_{min(l)}(k)\cup P_{max(l)}(k)$ such that $l_{ij}\left( k \right)>0$. This indicates that for $i\in P_{min(l)}(k)\cup P_{max(l)}(k)$, we get  $$x_{i(l)}(k+1)\in (x_{min(l)}(k), x_{max(l)}(k)).$$
	Agents not in  $P_{min(l)}(k)\cup P_{max(l)}(k)$ trivially satisfy the same inclusion, so   $\Xi (k+1)\subset \Xi (k)$  for all   $k\in[k_{s\,\,}, k_{s+1}).$

    \vspace{0.1cm}
	
    \textbf{Case 2).} For every  $i \in P_{min(l)}(k)$, the same reasoning as in Case 1) yields $x_{i(l)}(k+1)\in (x_{min(l)}(k), x_{max(l)}(k)), i\in P_{min(l)}(k)$.  For each agent $i \in P_{max(l)}(k)$, since  $N_i^{+}(k)\cap 
	P_{max(l)}(k)\neq\varnothing,~\exists i\in P_{max(l)}(k)$, the worst outcome  is  $x_{i(l)}(k+1)=  x_{max(l)}(k),~ i\in P_{min(l)}(k)$. Therefore, it is derived that   $\Xi (k+1)\subset \Xi (k)$ before achieving resilient consensus. As for the subcase  that $N_i^{+}(k)\cap 
	P_{max(l)}(k)=\varnothing,$  for all  $i\in P_{max(l)}(k)$ while    $N_i^{+}(k)\cap 
	P_{min(l)}(k)\neq\varnothing,~\exists i\in P_{min(l)}(k)$,
	then at least   one agent $ i\in P_{{max}(l)}(k)$ will be pulled strictly inside the interval, so  $\Xi (k_s+1)\subset \Xi (k_s)$.

    \vspace{0.1cm}
	
\textbf{Case 3).} Assume, for contradiction, that $\Xi (k+1)= \Xi (k)$ for $[\overline{k},+\infty)$. Then $x_{min(l)}(k)=x_{min(l)}(\overline{k})$ and $x_{max(l)}(k)=x_{max(l)}(\overline{k})$
for $k\in[\overline{k},+\infty)$,  which  further implies that  $P_{max(l)}(k)$ and  $P_{min(l)}(k)$ remain empty. 
While in alignment with the ASNS strategy,  $\mathcal{G}_{\mathcal{A}}(k)$ contains a spanning tree. Hence  some   agent ${i}$ in $\mathcal{A}(0,k)\backslash\left\{\tilde{i}\right\}$ has     in-neighbors   outside its own  set which are $P_{max(l)}(k)$ or $P_{min(l)}(k)$. Furthermore,  for    $(\ref{GOeqp71})$, since  $x_i\left( k+1 \right)$ is the linear combination of $x_i\left( k\right),~i\in \mathcal{A}(0,k)$ and $\epsilon \in (0,\mathop {\frac{1}{\max l_{ii}}})$, the cardinalities of   $P_{min(l)}(k)$ and  $P_{max(l)}(k)$ strictly decrease  until $x_{min(l)}(k)=x_{max(l)}(k)$, contradicting the assumption.

\vspace{0.1cm}
    
	To sum up, we have $\Xi (k+1)\subseteq \Xi (k)$ with  $ \Xi (k_s^{+})=\Xi (k_s^{-})$ and the equality  $\Xi (k+1)= \Xi (k)$ can persist only for a bounded time.   In this way, the  resilient consensus  is guaranteed, which completes the proof.
\end{proof}
\vspace{0.2cm}

\begin{remark}	
	The ASNS strategy  constructs a neighbor selection scheme  such that the resulting communication topology is  $p$-robust with $p \leqslant F+1$. This significantly relaxes the $(2F+1)$-robustness required by the time-invariant  topology in \cite{leblanc2013resilient}. Consequently, the approach reduces communication overhead while still ensuring consensus.
\end{remark}
\vspace{0.2cm}

\begin{remark}
	Unlike \cite{luo2023secure} and \cite{ZHAO2023110934},
	the ASNS strategy no longer presumes  that the underlying graph      among  normal agents should keep  the connection performance.  Instead, it actively builds  a  directed spanning tree. This design facilitates implementation, as the adversary's target behavior remains unknown.
\end{remark}

\vspace{0.1cm}

\begin{remark}
     The topology dynamics induced by the ASNS strategy  present greater challenges to  adversaries. Some sophisticated attacks, such as stealthy attacks \cite{10475155} and ripple attacks \cite{zhang111307Ripple}, rely on the topological information. The topology  dynamics of our work disrupts the adversaries' knowledge of the system model, thereby hindering the design of targeted attacks aligned with the system behavior.     
\end{remark}

\vspace{0.2cm}
Through the above analysis, it is evident that under the proposed ASNS  strategy, the communication cost  of network can be adjusted while maintaining resilience against attacks.      
Specifically, the number of  in-neighbors corresponding to each normal agent is adjusted  with $\psi_i(k)$, that is,  $N_{i}^{+}(k) \subseteq \psi_i(k)$. In other words, the communication remains flexible. Moreover, because communication overhead is often the dominant cost in real-world MASs, achieving resilience with the lowest possible data exchange is of paramount interest\cite{pirani2023graph}. Motivated by this, we evaluate the total defense cost of ASNS strategy  when communication is minimized.  We first give a formal definition of resilient minimum  communication in the presence of Byzantine agents, following the  idea  in  \cite{weerakkody2017robust}.
\vspace{0.2cm}
\begin{definition}\label{GOdf2}
	\emph{(Resilient Minimum Communication)} The MASs under $\mathcal{G} _{\mathcal{A}}\left( k \right)$  subject to   Byzantine attacks  are  said to achieve   resilient   minimum communication, if 
	$$\left| \mathbb{E} _{\mathcal{G} _{\mathcal{A}}}\left( k \right) \right|=\underset{g\in \mathbb{G}\left( k \right)}{\min }\left| \mathbb{E} _g\left( k \right) \right|,
	$$
	where $\mathbb{G}\left( k \right)$ is the set of  all the communication  graphs  for agents in $\mathcal{A}(0,k)$ that contain a directed spanning tree at time $k$ and  $\mathbb{E} _g\left( k \right)$ is the set of edges corresponding to graph $g$.
\end{definition}  
\vspace{0.2cm}
Next, the minimum communication overhead of the ASNS strategy is quantitatively analyzed. It is first noted that, according to \textit{Proposition}~\ref{GOth190}, under the ASNS strategy, each normal agent is guaranteed to have at least one selected in-neighbor. This structural property enables the exploration of defense mechanisms under minimum  communication cost.       

\vspace{0.1cm}

\begin{proposition}\label{GOthe30}
	Consider the  MASs $(\ref{GOeq499})$ with the  Byzantine attacks $(\ref{GOeq4117})$. Under  the ASNS strategy and  Assumption  $\ref{GOas02}$,  if {$\mathcal{G}_{pre}(0)$} is $(F+1)$-robust and all agents in $\mathcal{A}(0,k)$ except virtual leader  choose  $N_{i}^{+}\left( k \right) =\left\{ \left. j \right|j\in \psi _i\left( k \right) \right\}$ with $\left| N_{i}^{+}\left( k \right) \right|=1$,   $\mathcal{G}_{\mathcal{A}}(k)$  attains   resilient  minimum communication and    resilient consensus is  achieved.
\end{proposition}

\vspace{0.1cm}

\begin{proof}
	We proceed by contradiction to show that the graph $\mathcal{G}_{\mathcal{A}}(k)$ contains a spanning tree.
	Suppose that there is a non-empty subset $\varpi (k)$ of ${\mathcal{A}}(0,k) \backslash\{\tilde{i}\}$ that is unreachable  from agent $\tilde{i}$. Consider agent $i\in \varpi (k)$ with the smallest $v _{1(i)}\left( L_B\left( k \right) \right)$ among all agents in $\varpi (k)$. From  the  ASNS strategy,   agent $\tilde{i}$ is left with no selectable in-neighbors, i.e.,  $\psi_i(k) = \varnothing$, yielding    a contradiction.

\vspace{0.1cm}
    
	Next, because every agent in $\mathcal{A}(0,k)$ only chooses one in-neighbor from   $N_{i\text{-}pre}(0)\cap\mathcal{A}(0,k)$.  Thus, it is straightforward that $\mathcal{G}_{\mathcal{A}}(k)$ under the ASNS strategy satisfies resilient minimum communication. The resilient consensus can also be  realized based on  the poof in \textit{Theorem} \ref{GOth330}.
\end{proof}

\vspace{0.1cm}

\begin{remark}	
	Note that existing research primarily focuses on enhancing network communication  redundancy to improve resilience against Byzantine attacks~\cite{leblanc2013resilient,sundaram2018distributed,gong2023resilient}. The study in~\cite{weerakkody2017robust} investigates the minimum communication requirements under zero-dynamics attacks from the perspective of structural system theory. However, limited efforts have been devoted to leveraging minimum defense resources  to counteract Byzantine adversaries.  The proposed approach maintains strong resilience by adding new edges when the spanning tree among normal agents is disrupted, thereby mitigating the adverse effects on network connectivity.
\end{remark}

\vspace{-0.1cm}

\section{Simulations and Discussions}\label{GOsec5}

In this section, we first  elaborate on  the performance of the ASNS strategy. Next, comparative  simulations are carried out  to reveal the superiority  of our results.
\vspace{-0.2cm}
\subsection{Performance of ASNS Strategy under Byzantine attacks}
A directed graph of ten  agents whose initial graph containing  a directed spanning tree is considered. The  set of  compromised  agents is fixed at $\bar{\mathcal{B}}=\left\{ 1,4,9 \right\}$, which corresponds to an $F$-local Byzantine model with $F=2$. In practice, the set of  Byzantine agents  at time $k$, denoted by $\bigcup_{i\in \mathbb{V}} \mathcal{B}_i(k)$, is a subset of the predefined set $\bar{\mathcal{B}}$. For convenience, agents in $\bar{\mathcal{B}} \setminus \left( \bigcup_{i\in \mathbb{V}} \mathcal{B}_i(k) \right)$, which are not actively launching attacks at time $k$, are referred as  dormant Byzantine agents.

Fig. $\ref{GOfig1p0}$(a) depicts  the initial communication graph ${\mathcal{G}}(0)$ under the influence of the  Byzantine agents in $\bar{\mathcal{B}}$.  A key observation is that  the graph of  agents in $\bar{\mathcal{A}}$  has  no directed spanning tree. Consequently, the algorithm proposed in \cite{luo2023secure} becomes ineffective   when all agents in $\bar{\mathcal{B}}$ are compromised. In all simulations, we
set $\epsilon=0.02$.
Besides,  Fig. $\ref{GOfig1p0}$(b) displays   the initial 
pre-discriminative  graph  $\mathcal{G}_{pre}(0)$ which is $3$-robust. 

In fact, $\mathcal{G}_{pre}(0)$ specifies the set of admissible neighbors for all agents, delineating all potential communication links that can be established. The actual communication topology is a subgraph of $\mathcal{G}_{pre}(0)$. For example, ${\mathcal{G}}(0)$ in Fig.~\ref{GOfig1p0}(a) is a subgraph of $\mathcal{G}_{pre}(0)$ in Fig.~\ref{GOfig1p0}(b). It is also worth noting that $\mathcal{G}_{{pre}}(0)$ is not a complete graph; for instance, there is no edge between agents $1$ and $10$.

\begin{figure}[H]
	\vspace{-0.2cm}
	\centering
	\includegraphics[width=\linewidth]{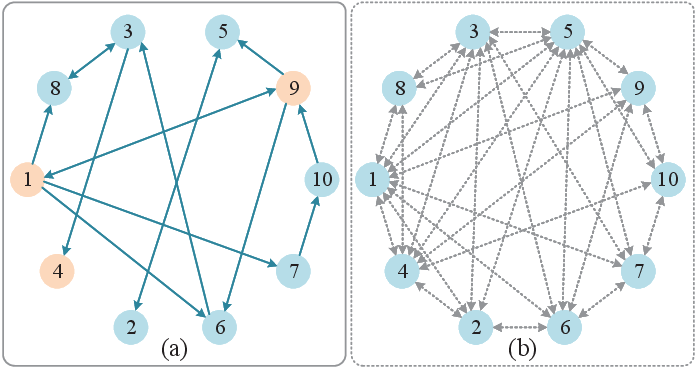}
	\caption{(a) The initial communication graph ${\mathcal{G}}(0)$; (b) The  initial pre-discriminative  graph $\mathcal{G}_{pre}(0)$.}
	\label{GOfig1p0}
\end{figure}

\vspace{-0.05cm}

The Byzantine attacks are designed as follows with  attack targets    changing  moments being    $k_1=120$ and $k_2=400$. The whole evaluation process is given below: 
 \begin{equation}\label{GOeq1a71}
	\begin{aligned}
		&f_{1j}(k)=\left\{\begin{array}{ll}B_{01},&j=8,k\in[120,400),
			\\B_{02},&j=6,k\in[120,400),\\x_1(k-6),&j=7,k\in[120,400),\\A_1x_1(k-3)+5  ,&j=9,k\in[120,400),\end{array}\right.\\               
            &f_{4j}(k)=\left\{\begin{array}{ll}B_{11},&j=2,k\in[400,\infty),
			\\A_1x_4(k)+B_{12},&j=10,k\in[400,\infty),\end{array}\right.\\
		&f_{9j}(k)=\left\{\begin{array}{ll}A_0x_9(k)+B_{01},&j=5,k\in[120,\infty),
			\\A_0x_9(k)+B_{02},&j=6,k\in[120,\infty),
			\\A_0x_9(k),&j=1,k\in[120,\infty).\end{array}\right.\\
	\end{aligned}
\end{equation}
where 
\begin{equation*}\label{GO23p871}
	\begin{aligned}
		\begin{cases}
			A_0=\mathrm{diag}\left\{ 0.03\sin \left( k \right) ,1,0.02\cos \left( k \right) \right\},\\
			A_1=\mathrm{diag}\left\{ 0.07\sin \left( k \right) ,1,0.02\sin \left( k \right) \right\},\\
		\end{cases}
	\end{aligned}
\end{equation*}
and 
\begin{equation*}\label{GO12p871}
	\begin{aligned}
		\begin{cases}
			B_{01}=\left[ \begin{matrix}
				0.02&		0.06&		0.04\\
			\end{matrix} \right] ^{\top},\\
			B_{02}=\left[ \begin{matrix}
				0.12&		0.36&		0.09\\
			\end{matrix} \right] ^{\top},\\
			B_{11}=\left[ \begin{matrix}
				0.12&		0.06&		0.26\\
			\end{matrix} \right] ^{\top},\\
			B_{12}=\cos \left( k \right) \left[ \begin{matrix}
				0.12&		0.36&		0.09\\
			\end{matrix} \right] ^{\top}.\\
		\end{cases}
	\end{aligned}
\end{equation*}

\begin{figure*}[htbp]
	\centering
	\includegraphics[width=0.85\linewidth]{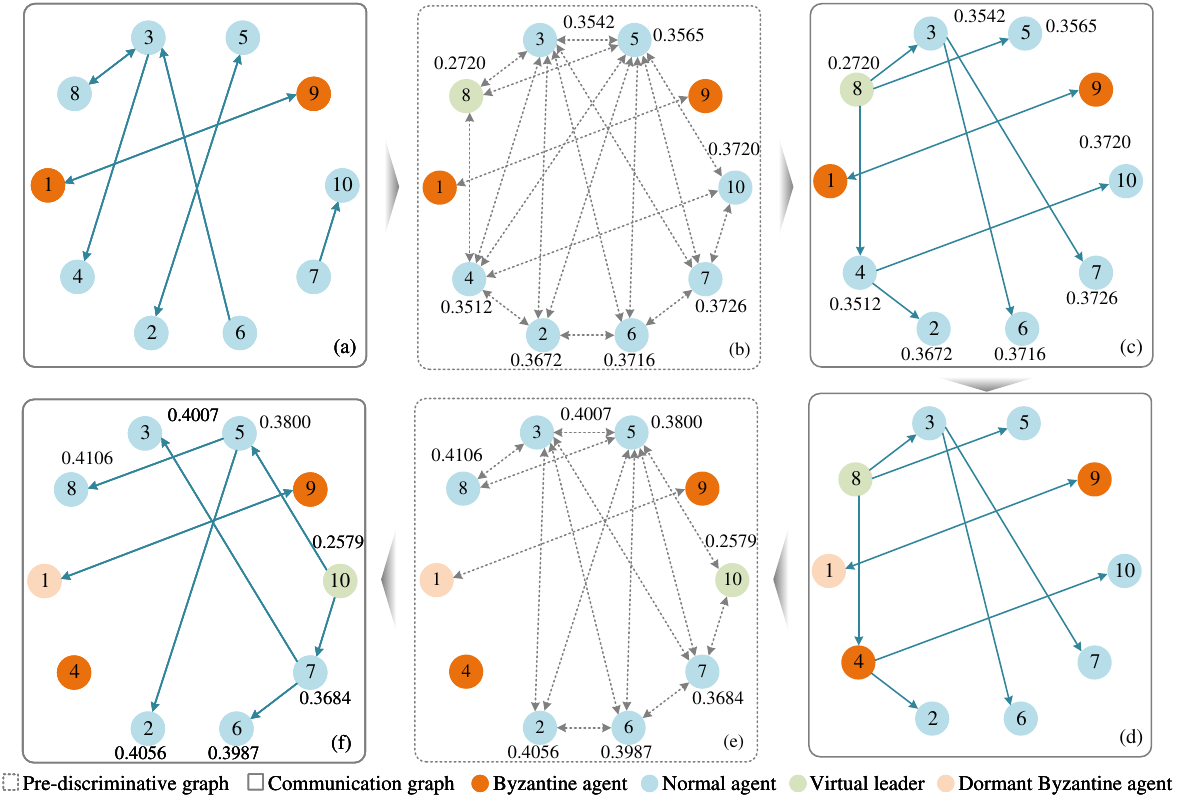}
	\caption{The communication graphs and pre-discriminative graphs: (a) $\mathcal{G}(k_1-1)$; (b) $\mathcal{G}_{pre}(k_1)$; (c) $\mathcal{G}(k_1)$; (d) $\mathcal{G}(k_2-1)$; (e) $\mathcal{G}_{pre}(k_2)$; (f) $\mathcal{G}(k_2)$.}
	\label{GOfig990}
\end{figure*}

$\bm{k=k_1=120}$: Based on the above attack model,  agents $1$ and $9$ are attacked as the Byzantine ones at $k_1=120$, see   Fig. $\ref{GOfig990}$(a). At  $k_1$,  in terms of  the ASNS strategy, Byzantine agents  $1$ and $9$ are isolated with   edges $(1,8)$,  $(1,7)$, $(1,6)$, $(10,9)$, $(9,5)$ and $(9,6)$ being deleted.  It is indicated that  $\mathcal{A} \left(0,k_1 \right) =\left\{ 2,3,4,5,6,7,8,10 \right\}$ such that $\mathcal{A} \left(0,k_1 \right) \neq \mathcal{A} \left(0,k_1-1\right)$.  Then a  pre-discriminative graph $\mathcal{G}_{pre}(k_1)$  is constructed  according to \textbf{\emph{Algorithm} \ref{al01e}} which is plotted in   Fig. $\ref{GOfig990}$(b) {(\textbf{Step 9}  in the ASNS strategy)}. 
The normal agent $8$ is chosen as a  virtual leader such that  $L_B(k_1)$ is formed. Then it follows that $v_1\left(L_B(k_1)\right)=[
0.3672~	0.3542~0.3512~0.3565~		0.3716~	0.3726~		0.2720~	0.3720
]$ (\textbf{Steps 11-12}  in the ASNS strategy).  The communication graph $\mathcal{G}(k_1)$ is then  reconstructed. Each agent $i$ in $\mathcal{A} \left(0,k_1 \right)$   selects at least one  in-neighbor as $N_{i}^{+}(k_1) \subseteq \psi_i(k_1) =\left\{ \left. j \right|v_{1(i)}\left( L_B\left( k_1 \right) \right) >v_{1(j)}\left( L_B\left( k_1 \right) \right) \right\} 
$. Then the new secure  communication graph  $\mathcal{G}(k_1)$ is rebuilt up as Fig. $\ref{GOfig990}$(c)  (\textbf{Steps 13-19}  in the ASNS strategy).
\vspace{0.0cm}

$\bm{k=k_2=400}$: Now the adversaries shift to agents $4$ and $9$.
The virtual leader is designated as agent $10$. The defense procedure is similar to  the above elaboration, which  is   shown in  Figs. $\ref{GOfig990}$(d)-(f). It is worthy to note  that isolating agent~4 disconnects agents~2 and~10 from the rest of the network (see Fig. \ref{GOfig990}(d)). Consequently, the method in~\cite{luo2023secure}, which relies on the connectivity assumption among normal agents, fails to achieve consensus under this condition.
\vspace{0.0cm}

Fortunately, with the help of ASNS strategy, the communication graphs  are rebuilt up  among normal agents. 
The relative error $\sigma _i(k)\triangleq  \left\| \sum_{j\in \bar{\mathcal{A}}}{\left( x_i\left( k \right) -x_j\left( k \right) \right)} \right\|,~i\in\bar{\mathcal{A}}$ and $\sigma _i(k)\triangleq  \left\| \sum_{j\in N_{i}^{+}(0)}{\left( x_i\left( k \right) -x_j\left( k \right) \right)} \right\|,~i\in\bar{\mathcal{B}}$ are provided to  quantify system performance which is illustrated in Fig. $\ref{GOfig9}$. It is found  that   the ASNS strategy mitigates the influence of adversaries and achieves consensus by dynamically forming a directed spanning tree.

\begin{figure}
	\centering
	\includegraphics[width=\linewidth]{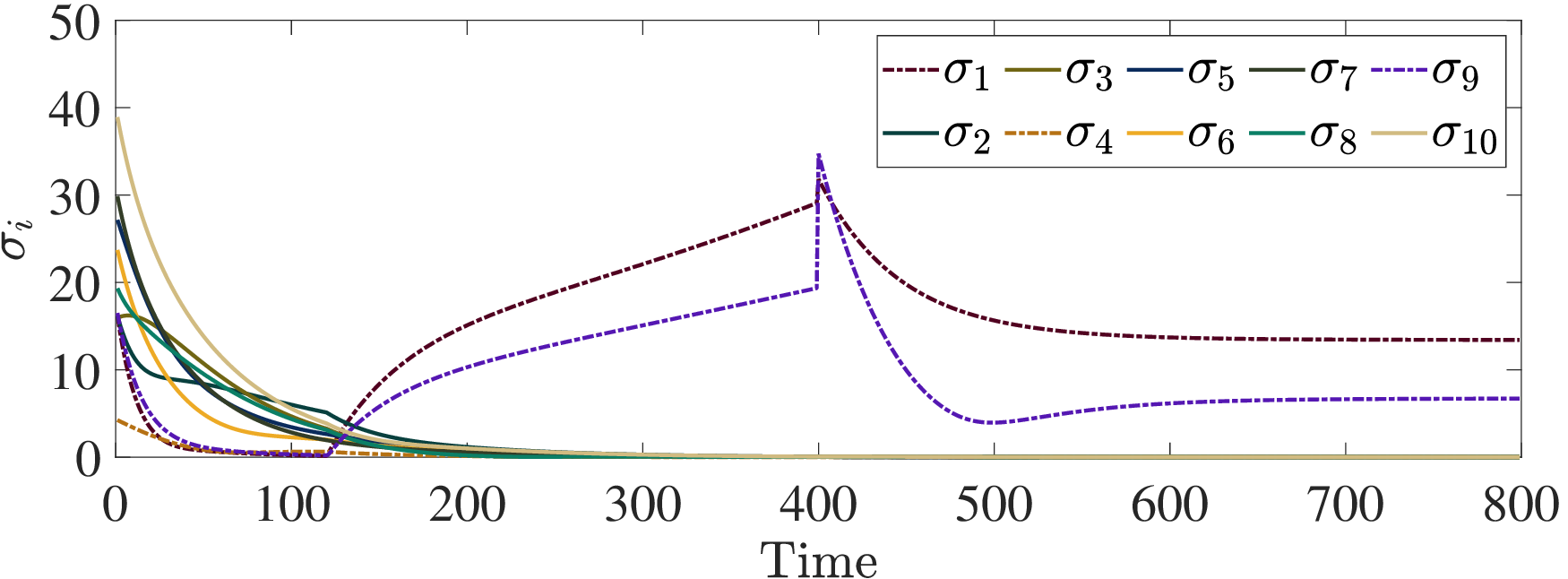}
	\caption{ The relative error of MASs under the 	ASNS strategy against Byzantine attacks.}
	\label{GOfig9}
\end{figure}

\subsection{Performance Comparison with Existing Results}
Consider the graph of Fig. $\ref{GOfig1p0}$(a) as the initial communication graph and the set of  Byzantine  agents is $\left\{ 1,9 \right\}$.  The attack strategies  of agents 1 and 9  are the  same as  $(\ref{GOeq1a71})$ where attacked time  periods   are both  $k\in[120,\infty)$. The removal of agent~9 results in the disconnection of agents~7 and~10 from the other normal agents (see Figs.~\ref{GOfig990}(a)). 

The ASNS strategy is first  contrasted with the method in~\cite{luo2023secure}, which depends on the connectivity assumption among normal agents. With this feature, the method in~\cite{luo2023secure} fails to achieve consensus under this condition. This is because the isolation of Byzantine agents undermines the communication topology of normal agents which  contains  no  directed spanning tree and results in insufficient interactions. Fig. $\ref{GOfig10}$ confirms this statement. However, the resilient consensus can still be achieved under the ASNS strategy by dynamically rebuilding  the communication graph  which is  depicted  in  Fig. $\ref{GOfig10_2}$.

Now, we compare  the ASNS strategy with the  W-MSR algorithms  \cite{leblanc2013resilient,sundaram2018distributed}.  As illustrated in the attack scenario, the Byzantine attack satisfies the $F$-local condition with $F=2$.
It is straightforward that  $\mathcal{G}(k_1-1)$ is 1-robust, not $(2F+1)$-robust  which indicates that the communication resources are insufficient for the W-MSR framework as elaborated in \cite{leblanc2013resilient,sundaram2018distributed}. 
In light of ASNS strategy, the  process of neighbor selection is similar to  the one  from  Fig. $\ref{GOfig990}$(a) to Fig. $\ref{GOfig990}$(c) and the resilient  consensus is satisfied  according to  Fig.  $\ref{GOfig10_2}$. To facilitate comparison, as shown in Fig.~\ref{GOfig12_2},  the W-MSR algorithm \cite{leblanc2013resilient,sundaram2018distributed}  is applied from  \( k = 80 \), in the absence of any  attacks. It is indicated that  even in a nominal setting, the interaction among agents  is disrupted, impeding convergence. Normal agents  fail to achieve resilient consensus under the W-MSR algorithm. It is  because   the network lacks the robustness required to resist attacks, and therefore cannot provide sufficient communication redundancy.

\section{Conclusion}\label{GOsec6}

An active neighbor selection strategy was presented via constructing  the pre-discriminative graph to ensure the consensus of MASs under Byzantine attacks. The flexible communication was  achieved by the adjustment of in-neighbor number. In this way, not only the resilient consensus is     guaranteed but also the communication resources can be  saved. Besides, the assumption about the connection performance among normal agents was released.   Furthermore, an  algorithm was proposed   to achieve the minimum number of edges within the normal agents while preserving  a directed spanning tree.

\begin{figure}[H]
	\centering
	\includegraphics[width=\linewidth]{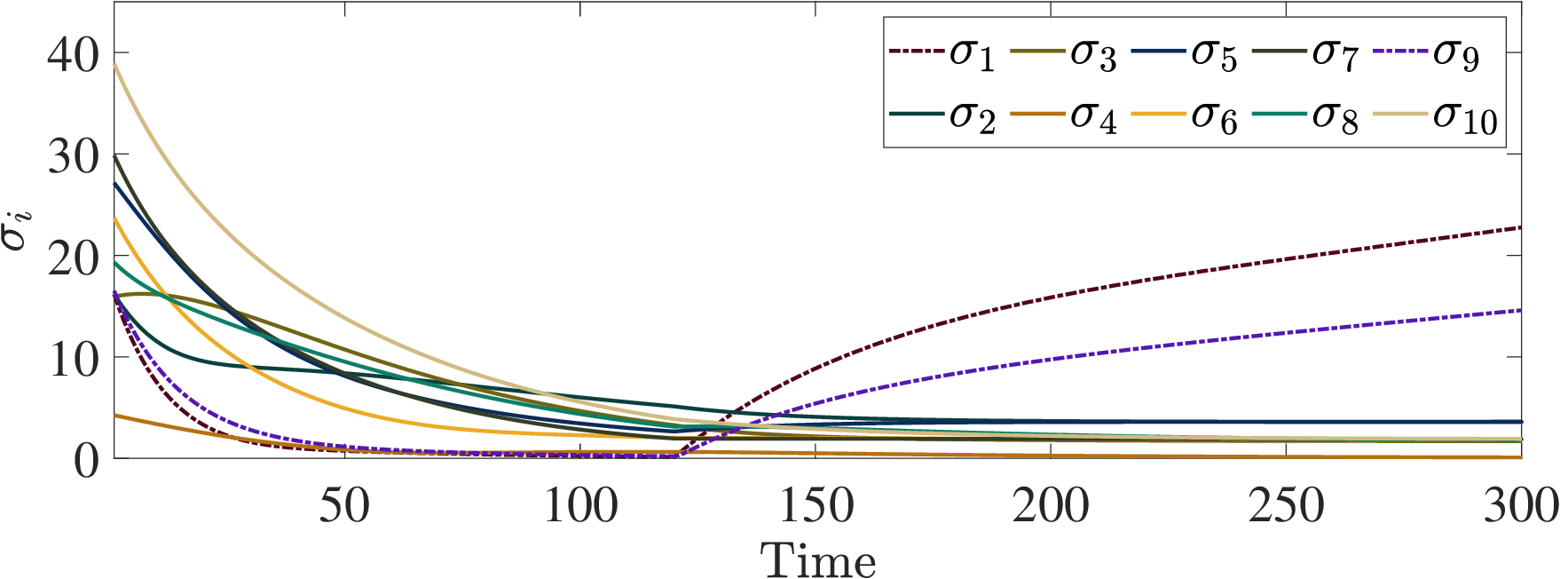}
	\caption{
The relative error of MASs under Byzantine attacks with the defense strategy in \cite{luo2023secure}  based on the connectivity-based assumption among normal agents .}
	\label{GOfig10}
\end{figure}

\begin{figure}[H]
	\centering
	\includegraphics[width=\linewidth]{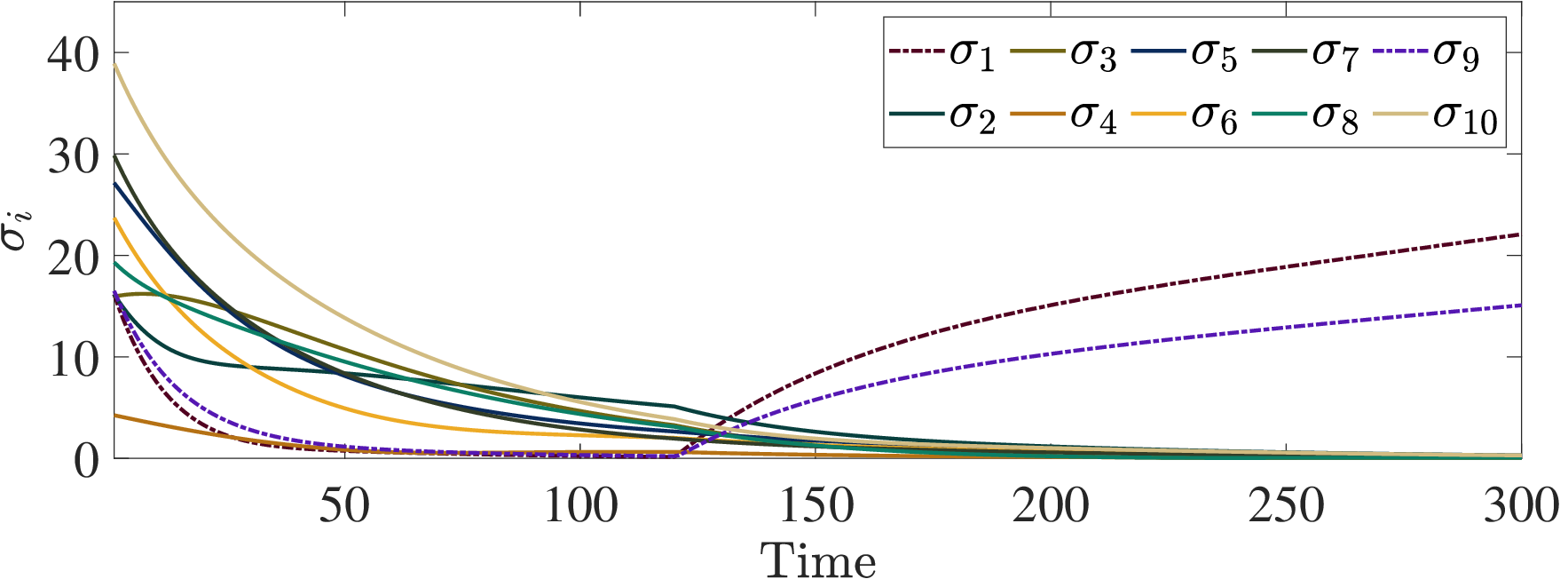}
	\caption{The relative error  of MASs suffering from  Byzantine attacks with the   ASNS strategy.}
	\label{GOfig10_2}
\end{figure}

\begin{figure}[H]
	\centering
	\includegraphics[width=\linewidth]{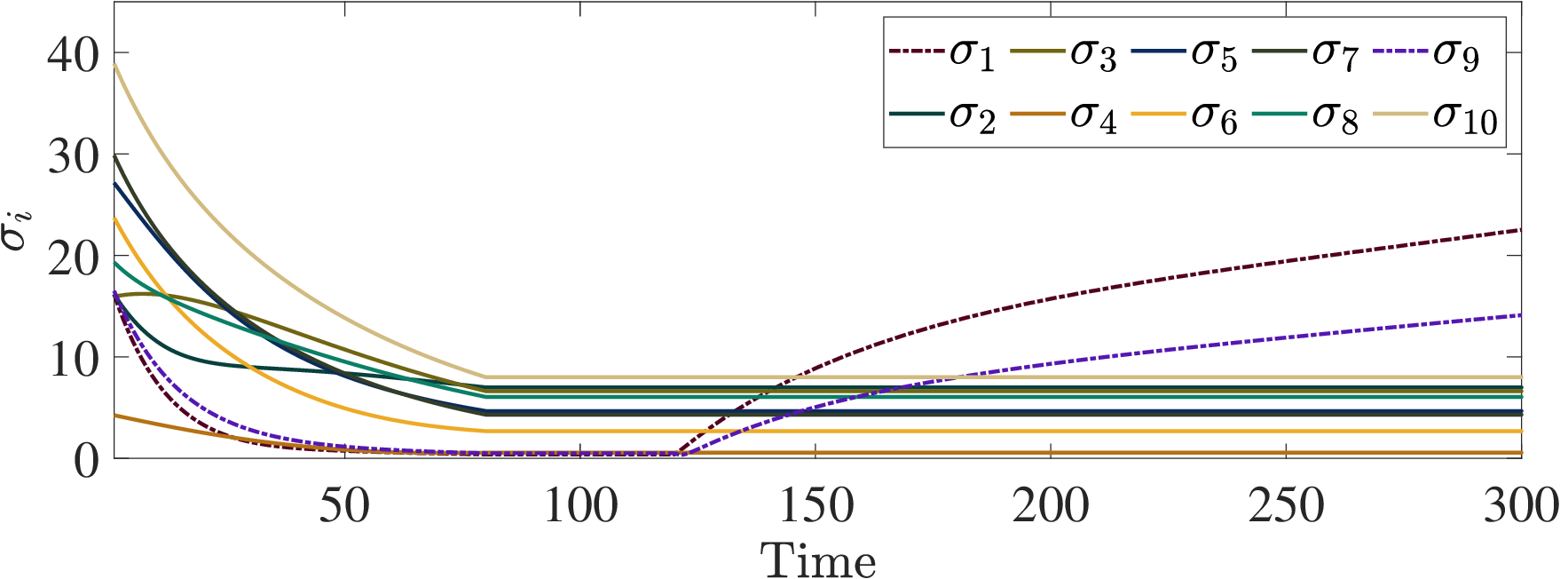}
	\caption{The relative error of MASs under Byzantine attacks with the W-MSR strategy requiring $(2F+1)$-robustness for resilient consensus \cite{leblanc2013resilient,sundaram2018distributed}.}
	\label{GOfig12_2}
\end{figure}

\bibliographystyle{ieeetr}
\bibliography{mybibfile}

\begin{thebibliography}{10}

\bibitem{103680}
S.~C. Hassler, U.~A. Mughal, and M.~Ismail, ``Cyber-physical intrusion
  detection system for unmanned aerial vehicles,'' {\em IEEE Transactions on
  Intelligent Transportation Systems}, vol.~25, no.~6, pp.~6106--6117, 2024.

\bibitem{watari2023duck}
D.~Watari, I.~Taniguchi, and T.~Onoye, ``Duck curve aware dynamic pricing and
  battery scheduling strategy using reinforcement learning,'' {\em IEEE
  Transactions on Smart Grid}, vol.~15, no.~1, pp.~457--471, 2023.

\bibitem{mo2015physical}
Y.~{Mo}, S.~{Weerakkody}, and B.~{Sinopoli}, ``Physical authentication of
  control systems: Designing watermarked control inputs to detect counterfeit
  sensor outputs,'' {\em IEEE Control Systems Magazine}, vol.~35, no.~1,
  pp.~93--109, 2015.

\bibitem{Maitra2015}
A.~K. Maitra, ``Offensive cyber-weapons: technical, legal, and strategic
  aspects,'' {\em Environment Systems and Decisions}, vol.~35, no.~1,
  pp.~169--182, 2015.

\bibitem{pasqualetti2013attack}
F.~Pasqualetti, F.~D{\"o}rfler, and F.~Bullo, ``Attack detection and
  identification in cyber-physical systems,'' {\em IEEE transactions on
  automatic control}, vol.~58, no.~11, pp.~2715--2729, 2013.

\bibitem{Gallo2020}
A.~J. Gallo, M.~S. Turan, F.~Boem, T.~Parisini, and G.~Ferrari-Trecate, ``A
  distributed cyber-attack detection scheme with application to {DC}
  microgrids,'' {\em IEEE Transactions on Automatic Control}, vol.~65, no.~9,
  pp.~3800--3815, 2020.

\bibitem{pirani2023graph}
M.~Pirani, A.~Mitra, and S.~Sundaram, ``Graph-theoretic approaches for
  analyzing the resilience of distributed control systems: A tutorial and
  survey,'' {\em Automatica}, vol.~157, p.~111264, 2023.

\bibitem{zhao2023active}
R.~Zhao, Z.~Zuo, Y.~Wang, and W.~Zhang, ``Active control strategy for switched
  systems against asynchronous {DoS} attacks,'' {\em Automatica}, vol.~148,
  p.~110765, 2023.

\bibitem{lu2023distributed}
A.-Y. Lu and G.-H. Yang, ``Distributed secure state estimation for linear
  systems against malicious agents through sorting and filtering,'' {\em
  Automatica}, vol.~151, p.~110927, 2023.

\bibitem{leblanc2013resilient}
H.~J. LeBlanc, H.~Zhang, X.~Koutsoukos, and S.~Sundaram, ``Resilient asymptotic
  consensus in robust networks,'' {\em IEEE Journal on Selected Areas in
  Communications}, vol.~31, no.~4, pp.~766--781, 2013.

\bibitem{olfati2007consensus}
R.~Olfati-Saber, J.~A. Fax, and R.~M. Murray, ``Consensus and cooperation in
  networked multi-agent systems,'' {\em Proceedings of the IEEE}, vol.~95,
  no.~1, pp.~215--233, 2007.

\bibitem{chen2024cooperative}
F.~Chen, M.~Sewlia, and D.~V. Dimarogonas, ``Cooperative control of
  heterogeneous multi-agent systems under spatiotemporal constraints,'' {\em
  Annual Reviews in Control}, vol.~57, p.~100946, 2024.

\bibitem{10014016}
M.~Luo, B.~Du, W.~Zhang, T.~Song, K.~Li, H.~Zhu, M.~Birkin, and H.~Wen, ``Fleet
  rebalancing for expanding shared e-mobility systems: A multi-agent deep
  reinforcement learning approach,'' {\em IEEE Transactions on Intelligent
  Transportation Systems}, vol.~24, no.~4, pp.~3868--3881, 2023.

\bibitem{10018476}
Z.~Fan, W.~Zhang, and W.~Liu, ``Multi-agent deep reinforcement learning-based
  distributed optimal generation control of {DC} microgrids,'' {\em IEEE
  Transactions on Smart Grid}, vol.~14, no.~5, pp.~3337--3351, 2023.

\bibitem{zheng2018average}
M.~Zheng, C.-L. Liu, and F.~Liu, ``Average-consensus tracking of sensor network
  via distributed coordination control of heterogeneous multi-agent systems,''
  {\em IEEE Control Systems Letters}, vol.~3, no.~1, pp.~132--137, 2018.

\bibitem{zhang2022much}
W.~Zhang, Z.~Zuo, Y.~Wang, and G.~Hu, ``How much noise suffices for privacy of
  multiagent systems?,'' {\em IEEE Transactions on Automatic Control}, vol.~68,
  no.~10, pp.~6051--6066, 2022.

\bibitem{zuo2021resilient}
Z.~Zuo, X.~Cao, Y.~Wang, and W.~Zhang, ``Resilient consensus of multiagent
  systems against denial-of-service attacks,'' {\em IEEE Transactions on
  Systems, Man, and Cybernetics: Systems}, vol.~52, no.~4, pp.~2664--2675,
  2021.

\bibitem{zegers2021event}
F.~M. Zegers, M.~T. Hale, J.~M. Shea, and W.~E. Dixon, ``Event-triggered
  formation control and leader tracking with resilience to {Byzantine}
  adversaries: A reputation-based approach,'' {\em IEEE Transactions on Control
  of Network Systems}, vol.~8, no.~3, pp.~1417--1429, 2021.

\bibitem{mustafa2020resilient}
A.~Mustafa, H.~Modares, and R.~Moghadam, ``Resilient synchronization of
  distributed multi-agent systems under attacks,'' {\em Automatica}, vol.~115,
  p.~108869, 2020.

\bibitem{yuan2021secure}
L.~Yuan and H.~Ishii, ``Secure consensus with distributed detection via two-hop
  communication,'' {\em Automatica}, vol.~131, p.~109775, 2021.

\bibitem{luo2023secure}
X.~Luo, C.~Zhao, and J.~He, ``Secure multi-dimensional consensus algorithm
  against malicious attacks,'' {\em Automatica}, vol.~157, p.~111224, 2023.

\bibitem{ishii2022overview}
H.~Ishii, Y.~Wang, and S.~Feng, ``An overview on multi-agent consensus under
  adversarial attacks,'' {\em Annual Reviews in Control}, vol.~53,
  pp.~252--272, 2022.

\bibitem{8795564}
J.~Usevitch and D.~Panagou, ``Resilient leader-follower consensus to arbitrary
  reference values in time-varying graphs,'' {\em IEEE Transactions on
  Automatic Control}, vol.~65, no.~4, pp.~1755--1762, 2020.

\bibitem{usevitch2020determining}
J.~Usevitch and D.~Panagou, ``Determining r-and (r, s)-robustness of digraphs
  using mixed integer linear programming,'' {\em Automatica}, vol.~111,
  p.~108586, 2020.

\bibitem{sundaram2018distributed}
S.~Sundaram and B.~Gharesifard, ``Distributed optimization under adversarial
  nodes,'' {\em IEEE Transactions on Automatic Control}, vol.~64, no.~3,
  pp.~1063--1076, 2018.

\bibitem{10354416}
M.~Cavorsi, L.~Sabattini, and S.~Gil, ``Multirobot adversarial resilience using
  control barrier functions,'' {\em IEEE Transactions on Robotics}, vol.~40,
  pp.~797--815, 2024.

\bibitem{ZHAO2023110934}
D.~Zhao, Y.~Lv, G.~Wen, and Z.~Gao, ``Resilient consensus of high-order
  networks against collusive attacks,'' {\em Automatica}, vol.~151, p.~110934,
  2023.

\bibitem{an2024mean}
L.~An and G.-H. Yang, ``Mean-square exponential convergence for
  {Byzantine}-resilient distributed state estimation,'' {\em Automatica},
  vol.~163, p.~111592, 2024.

\bibitem{shao2023distributed}
H.~Shao, L.~Pan, M.~Mesbahi, Y.~Xi, and D.~Li, ``Distributed neighbor selection
  in multiagent networks,'' {\em IEEE Transactions on Automatic Control},
  vol.~68, no.~11, pp.~6711--6726, 2023.

\bibitem{DION20031125}
J.-M. Dion, C.~Commault, and J.~{van der Woude}, ``Generic properties and
  control of linear structured systems: a survey,'' {\em Automatica}, vol.~39,
  no.~7, pp.~1125--1144, 2003.

\bibitem{ren2008distributed}
W.~Ren and R.~W. Beard, {\em Distributed consensus in multi-vehicle cooperative
  control}.
\newblock London, U.K.: Springer, 2008.

\bibitem{yuan2024resilient}
L.~Yuan and H.~Ishii, ``Resilient average consensus with adversaries via
  distributed detection and recovery,'' {\em IEEE Transactions on Automatic
  Control}, vol.~70, no.~1, pp.~415--430, 2025.

\bibitem{YUAN2025111908}
L.~Yuan and H.~Ishii, ``Asynchronous approximate {Byzantine} consensus: A
  multi-hop relay method and tight graph conditions,'' {\em Automatica},
  vol.~171, p.~111908, 2025.

\bibitem{10102299}
L.~Yuan and H.~Ishii, ``Event-triggered approximate {Byzantine} consensus with
  multi-hop communication,'' {\em IEEE Transactions on Signal Processing},
  vol.~71, pp.~1742--1754, 2023.

\bibitem{yan2022resilient}
J.~Yan, X.~Li, Y.~Mo, and C.~Wen, ``Resilient multi-dimensional consensus in
  adversarial environment,'' {\em Automatica}, vol.~145, p.~110530, 2022.

\bibitem{10475155}
H.~Guo, Z.-H. Pang, and C.~Li, ``Side information-based stealthy false data
  injection attacks against multi-sensor remote estimation,'' {\em IEEE/CAA
  Journal of Automatica Sinica}, vol.~11, no.~4, pp.~1054--1056, 2024.

\bibitem{zhang111307Ripple}
T.-Y. Zhang, D.~Ye, and G.-H. Yang, ``Ripple effect of cooperative attacks in
  multi-agent systems: Results on minimum attack targets,'' {\em Automatica},
  vol.~159, p.~111307, 2024.

\bibitem{weerakkody2017robust}
S.~Weerakkody, X.~Liu, and B.~Sinopoli, ``Robust structural analysis and design
  of distributed control systems to prevent zero dynamics attacks,'' in {\em
  2017 IEEE 56th Annual Conference on Decision and Control (CDC)},
  pp.~1356--1361, IEEE, 2017.

\bibitem{gong2023resilient}
X.~Gong, X.~Li, Z.~Shu, and Z.~Feng, ``Resilient output formation-tracking of
  heterogeneous multiagent systems against general {Byzantine} attacks: A
  twin-layer approach,'' {\em IEEE Transactions on Cybernetics}, vol.~54,
  no.~4, pp.~2566--2578, 2023.

\end{thebibliography}

\end{document}